\documentclass[journal]{IEEEtran}
\ifCLASSINFOpdf
\else
\fi
\usepackage{graphicx} 
\usepackage{amsmath} 
\usepackage{amssymb}  
\usepackage{amsthm}
\usepackage{xcolor}
\usepackage{cite}
\usepackage{caption}
\newtheorem{theorem}{Theorem}
\newtheorem{proposition}[theorem]{Proposition}
\newtheorem{lemma}[theorem]{Lemma}
\newtheorem{corollary}[theorem]{Corollary}
{}
\usepackage[ruled,noend]{algorithm2e}

\usepackage[colorlinks=true,allcolors=steelblue]{hyperref}
\definecolor{steelblue}{RGB}{70,130,180}
\title{\LARGE \bf
Higher-Order Moment-Based Anomaly Detection}

\author{Venkatraman Renganathan, Navid Hashemi, Justin Ruths, and Tyler H. Summers
 \thanks{This work was partially supported by the Army Research Office and was accomplished under Grant Number: W911NF-17-1-0058 and the Air Force Office of Scientific Research under award number FA2386-19-1-4073.}
\thanks{The authors are with the Department of Mechanical Engineering, The University of Texas at Dallas, 800 W. Campbell Rd, Richardson, TX, USA. Email: {\tt\small (vrengana, navid.hashemi, jruths, tyler.summers)@utdallas.edu}}%
}
 
\begin{document}

\maketitle
\thispagestyle{empty}
\pagestyle{empty}

\begin{abstract}
The identification of anomalies is a critical component of operating complex, large-scale and geographically distributed cyber-physical systems. While designing anomaly detectors, it is common to assume Gaussian noise models to maintain tractability; however, this assumption can lead to the actual false alarm rate being significantly higher than expected. Here we design a distributionally robust threshold of detection using finite and fixed higher-order moments of the detection measure data such that it guarantees the actual false alarm rate to be upper bounded by the desired one. Further, we bound the states reachable through the action of a stealthy attack and identify the trade-off between this impact of attacks that cannot be detected and the worst-case false alarm rate. Through numerical experiments, we illustrate how knowledge of higher-order moments results in a tightened threshold, thereby restricting an attacker's potential impact.
\end{abstract}

\section{INTRODUCTION} \label{sec_intro}
From critical infrastructures and industrial process control to autonomous driving and various biomedical applications, dynamical control systems are increasingly able to be instrumented with new sensing and actuation capabilities. These cyber-physical systems (CPS) comprise growing webs of interconnected feedback loops and must operate efficiently and resiliently in dynamic and uncertain environments. As these systems become large, devising both model-based \cite{murguia2019model,umsonst2018anomaly} and data-driven \cite{renganathan2020distributionally,khojasteh2020learning,martinez_tuning} methods for detecting anomalies (such as component failures or malicious attacks) are critical for their robust and efficient operation. Such critically important cyber-physical networks have become an attractive target to attackers. These systems are large and complex and are often not monitored well enough, enabling attackers to manipulate the system without being detected and cause damage \cite{zetter2014countdown,cherepanov2017industroyer}.

To simplify the analysis and design, often such complex cyber-networks are modeled as a discrete-time linear time invariant system with Gaussian noises. However, this can lead to a significant miscalculation of probabilities and risk if the underlying processes behave differently, for example due to various nonlinearities or malicious attacks. In the context of attacks, it is possible for an attacker to modify the sensor outputs and effectively generate aggressive and strategic noise profiles to sabotage the operation of the system. With stochastic optimization techniques, particularly using the emerging area of distributionally robust optimization (DRO) approaches \cite{dr_wiesemann}, these limitations can be recognized and addressed. DRO enables modelers to explicitly incorporate inherent ambiguity in probability distributions into optimization problems. DRO approaches can be categorized based on the form of the ambiguity set. There are several different parameterizations, including those based on moments, support, directional derivatives \cite{dr_goh}, and Wasserstein balls \cite{dr_peyman}. In practice, we have access to only a finite amount of historical data. However, it is possible to use finite historical data and guarantee resiliency in such critical cyber-physical networks. Here, we propose to use moment-based DRO methods to improve modeling and reduce false alarm rates in cyber-physical networks. 

In the context of attacks, the detector tuning has a direct implication on the effect an attacker can have while still remaining stealthy \cite{navid,umsonst2018anomaly}. A model-based approach to attack detection uses a detector that raises alarms when there is a large enough discrepancy between the actual and predicted measurements, a statistic termed the residual. The detector's sensitivity can be increased by decreasing the threshold of detection, but there is an inherent trade-off between sensitivity and the rate at which false alarms are generated. Keeping false alarms to a manageable level requires adjusting sensitivity and the tuning of the detector threshold is typically informed by the distribution of the residual. Authors in \cite{navid_GMM} used Gaussian Mixture Model to approximate the arbitrary noise distributions and obtained a detector threshold corresponding to a desired false alarm rate. 

However, when noise distributions are only known to an ambiguity set, traditional tools and approximations no longer suffice to select the threshold and so we turn to a distributionally robust approach. Interest in a DRO-informed perspective on detector tuning is supported by recent work on using a Wasserstein metric \cite{martinez_tuning}. Our moment-based ambiguity set formulation includes all distributions with fixed moments up to some order. The problem of designing anomaly detector thresholds subject to moment constraints of system uncertainties can be addressed using Generalized Moment Problems described in \cite{lasserre2010moments, lasserre2001global}. Further, the conditions for a truncated (finite) moment sequence to represent a probability measure were studied in \cite{curto2008extremal, helton2012semidefinite}. Authors in \cite{bertsimas_popescu} proposed semidefinite programs to compute a probability bound for a random variable lying in a set with known moments up to some order. These techniques can be utilized to design detector thresholds for residual distributions consistent with finite fixed moments up to some order. 

\textit{Contributions:} This paper is a significant extension of our previous work \cite{renganathan2020distributionally} where we used a moment-based ambiguity set formulation with fixed first two moments (Proposition \ref{prop_dr_alfa}) to obtain a detector threshold via generalized Chebyshev inequality. The main contributions of the present paper are:
\begin{itemize}
    \item We propose an approach to construct moment-based ambiguity sets with fixed moments up to $k\textsuperscript{th}$ order for the anomaly detection measure and design an anomaly detector threshold for CPSs that exhibit non-Gaussian uncertainties. The approach can utilize either residual moments obtained through a dynamic model or estimated directly from residual data.
    \item We use a semidefinite program (SDP) defined using higher-order moments of the detection measure to find a sharper probability bound for classifying the residual with an improved detector threshold (Theorem \ref{thm_alfa_higher_moments}).
    \item We illustrate empirically that inclusion of higher-order moments results in tightened threshold (Lemma \ref{lemma_1}), thereby restricting an attacker's potential impact, and also prove that the volume of the attack-reachable set shrinks with a tightened threshold (Corollary  \ref{cor_reach_shrinks}).
\end{itemize} 
While anomaly detection is widely studied in CPS literature, our distributionally robust approach using higher-order moments of the detection measure data marks the novel contribution of this paper. The rest of the paper is organized as follows. Section \ref{sec_momemt_problem} formulates the problem using moment-based ambiguity sets. In Section \ref{sec_dr_threshold_design}, the design of anomaly detector threshold is discussed. Section \ref{sec_reach_sets} describes the procedure to find the boundary of the reachable sets obtained using distributionally robust tuned detector. Section \ref{sec_numerical_simulation} presents the numerical results with inferences. Finally, Section \ref{sec_conclusion} concludes and summarizes the future research directions.
\section{Problem Formulation Using Higher-Order Moment-Based Ambiguity Sets} \label{sec_momemt_problem}
In this section, we propose a framework for designing an anomaly detector threshold for cyberphysical systems. The approach can utilize either model-based propagation of residual moments or data-driven estimation of detection measure moments directly from data. 
\subsection{Model-based Problem Formulation}
Here, we model an uncertain cyber-physical system as a stochastic discrete-time linear system
\begin{align} \label{eqn_uncertain_cps}
    x_{t+1} &= A x_t + B u_t + w_t, \quad t \in \mathbb{N} \\
    y_t &= C x_t + v_t,
\end{align}
where $x_t \in \mathbb{R}^n, u_t \in \mathbb{R}^m$ are the system state and input respectively at time $t$. The matrices $A$ and $B$ denote the system matrix and control input matrix, respectively. The output $y_t \in \mathbb{R}^p$ aggregates a linear combination of the states with the observation matrix $C \in \mathbb{R}^{p \times n}$. We assume that the pair $(A,C)$ is detectable and $(A,B)$ is stabilizable. The process noise $w_t \in \mathbb{R}^{n}$ and the sensor noise $v_t \in \mathbb{R}^{p}$ are modeled as zero-mean random vectors independent and identically distributed across time with covariance matrix $\Sigma_{w}, \Sigma_{v}$ respectively. Let $\kappa = (k_1,\dots,k_n)^{\top}$ with $k_j \in \mathbb{Z}_{+}$ non-negative integers and $J_k = \{ \kappa | \sum^{n}_{i = 1} k_i \leq k \}$. Further, assume that the feasible first $k$ moment sequence of the distributions of the $w_t, v_t$ namely $M^{\kappa}_{w}, M^{\kappa}_{v}, \kappa \in J_k$ respectively are known. The distributions $P_w$ of $w_t$ and $P_v$ of $v_t$ are unknown (and not necessarily Gaussian and possibly heavy-tailed\footnote{For the purposes of this paper, we consider heavy-tailed distributions as those whose moments above a certain order may be infinite, in which case their tails are heavier than a Gaussian. We assume the moments up to order $k$ are finite.}) and will be assumed to belong to the $k$-moments-based ambiguity sets of distributions $\mathcal{P}^{w}_{k}$ and $\mathcal{P}^{v}_{k}$ respectively defined as follows
\begin{align} \label{eqn_ambig_set_w_high}
    \mathcal{P}^{w}_{k} = \{ P_{w} \mid \mathbb{E}[w^{k_1}_{t_1} w^{k_2}_{t_2} \dots w^{k_n}_{t_n} ] = M^{\kappa}_{w}, \ \kappa \in J_k \}, \\ 
   \mathcal{P}^{v}_{k} = \{ P_{v} \mid \mathbb{E}[v^{k_1}_{t_1} v^{k_2}_{t_2} \dots v^{k_p}_{t_p} ] = M^{\kappa}_{v}, \ \kappa \in J_k \}  \label{eqn_ambig_set_v_high}.
\end{align}
When the actual measurement $y_t$ is corrupted by an additive attack, $\delta_t \in \mathbb{R}^p$, the true output of the system fed to the controller becomes
\begin{align}
    \Bar{y}_t = y_t + \delta_t = C x_t + v_t + \delta_t.
\end{align}
We utilize a steady-state Kalman filter to construct a state estimate $\hat{x}_{t}$ to minimize the squared norm of the estimation error $e_t = x_t - \hat{x}_t$ where,
\begin{align}
    \hat{x}_{t+1} &= A \hat{x}_t + B u_t + L (\bar{y}_t - C \hat{x}_t), 
\end{align}
and the estimation error evolves as
\begin{align}\label{eqn_error_dynamics}
    e_{t+1} &= (A - LC) e_t + w_t - L v_t - L \delta_t. 
\end{align}
In the absence of attacks (that is, $\delta_t = 0$) and when the covariance matrices of the noises are fixed and known, the Kalman gain $L = PC^{\top}(CPC^{\top} + \Sigma_{v})^{-1}$ minimizes the steady state covariance matrix 
\begin{align}
P &:= \lim_{t \rightarrow \infty} P_t := \mathbf{E}[e_t e^{\top}_t].
\end{align}
Since $(A,C)$ is assumed to be detectable, the existence of $P$ is guaranteed and it can be found through the solution of an algebraic Ricatti equation. We define a residual sequence $r_t$ as the difference between the actual received output $\Bar{y}_t$ and the predicted output $C \hat{x}_t$ as, 
\begin{align} \label{eqn_r}
    r_{t} &= \bar{y}_t - C \hat{x}_t = C e_t + v_t + \delta_t,
\end{align}
and in the attack free setting, $r_t$ falls according to a zero mean distribution with covariance 
\begin{align} \label{eqn_residual}
    \Sigma_{r} = \mathbf{E}[r_t r^{\top}_t] = C P C^{\top} + \Sigma_{v}.
\end{align}
Since \eqref{eqn_r} is linear, it is possible to obtain the fixed and first $k$ moments of the random variable $r_{t}$ by propagating the corresponding moments of the primitive random variables $w_t, v_t$. Thus, the distribution $P_r$ of $r_t$ (not necessarily Gaussian) belongs to an ambiguity set $\mathcal{P}^{r}_{k}$ given by 
\begin{align} \label{eqn_ambig_set_r}
    \mathcal{P}^{r}_{k} = \{ P_{r} \mid \mathbb{E}[r^{k_1}_{t_1} r^{k_2}_{t_2} \dots r^{k_p}_{t_p} ] = M^{\kappa}_{r}, \ \kappa \in J_k \}.
\end{align} 
We define detection measure $q_t$ as a quadratic function of $r_t$
\begin{align} \label{eqn_q}
    q_t = r^{\top}_t \Sigma^{-1}_{r} r_t,
\end{align}
which will be compared to a threshold for anomaly detection.

\subsection{Data-Driven Moment Estimation from Residual Data}
An alternative to obtaining residual moments by propagating the moments of primitive random variables through the system model is to instead collect residual data $r_t$ (from attack-free operation) and estimate residual moments or the moments of $q_t$ directly from the data. Such a data-driven approach allows our proposed tuning approaches to be used in much broader settings where it is difficult to propagate moments through a model (or even to obtain a model), but where residual data is easily generated from sensors and a state estimator. Higher-order moments require increasingly more data to obtain accurate estimates. Determining the required amount of data is possible using finite-sample measure concentration results given as in  \cite{shawe2003estimating}, but we leave such an analysis for future work. The proposed methods for setting thresholds for anomaly detection can be used together with ambiguity sets built upon data-driven residual estimates, although the false alarm rates will also be affected by sampling errors. Moment estimation uncertainty could be accommodated using the same generalized moment problem computations we propose here, but with assumed \emph{bounds} on moment estimates rather than having them fixed to exact known values. Such problems can also be reformulated in a computationally tractable manner, as described in, e.g., Section 3 of \cite{delage2010distributionally}, which incorporate uncertainty sets for estimated moments. These formulations would simply impose further constraints on our primal problem defined in subsection III.B of our revised manuscript. It is possible and would be interesting to explore how to use statistical confidence intervals for data-driven moment estimates to inform bounds used in the SDP and obtain end-to-end guarantees on false alarm rates associated with certain thresholds. This will be pursued in future work. 

In either setting, the \textit{feasible} first $k$ moment sequence of the distribution of $q_t$ denoted by $M^{k}_{q}$ is assumed to be known either from its primitive variables through the model-based approach or estimated from data. Then, the $k$-moments-based ambiguity set of the scalar random variable $q_t$ is defined as
\begin{align} \label{eqn_ambig_set_q_high}
    \mathcal{P}^{q}_{k} = \{ P_{q} \mid  \mathbb{E}[q^{k}_{t}] = M^{k}_{q} \}.
\end{align}
\begin{proposition} \label{def_rs_reach}
\cite{bertsimas_popescu} Consider a univariate random variable $X$ defined on $\Omega=\mathbb{R}_{+}$, endowed with its Borel sigma algebra of events. A sequence $\bar{\sigma} = (M_1, M_2, \dots, M_k)^{\top}$ is a \textbf{feasible} $(1,k,\Omega)-$moment vector of the random variable $X$, if and only if the matrices $R_k \succeq 0$ and $R_{k-1} \succeq 0$, where for any integer $l \geq 0$, the matrices are defined as
\begin{align} \label{eqn_mom_feas_cdtn}
\text {\scriptsize
    $R_{2l} = \begin{bmatrix} 1 & M_1 & \dots & M_l \\ M_1 & M_2 & \dots & M_{l+1} \\ \vdots & \vdots & \ddots & \vdots \\ M_l & M_{l+1} & \dots & M_{2l} \end{bmatrix}, R_{2l+1} = \begin{bmatrix} M_1 & \dots & M_{l+1} \\ M_2 & \dots & M_{l+2} \\ \vdots & \ddots & \vdots \\ M_{l+1} & \dots & M_{2l+1}
    \end{bmatrix}.$}
\end{align}
\end{proposition}

\section{Design of Anomaly Detector Thresholds} \label{sec_dr_threshold_design}
Given a threshold \footnote{The first and second subscripts in the threshold separated by a comma denote the random variable and number of moments respectively.}$\alpha_{q,k} > 0$ and the distance measure $q_t$, alarm time(s) $t^{\star}$ are produced according to the following rules,

\begin{align} \label{eqn_detectorthresh}
    \begin{cases} q_t \leq \alpha_{q,k}, &\text{no alarm} \\
    q_t > \alpha_{q,k}, &\text{alarm: } t^{\star} = t.
    \end{cases}
\end{align}
Even in the absence of attacks, the detector is expected to generate false alarms due to the infinite support of $v_t$, because some values drawn from $P_{q}$ will exceed the threshold $\alpha_{q,k}$. If $P_{q}$ is known, then it is possible to extract an optimum threshold value $\alpha_{q}^*$ from the corresponding cumulative distribution function $F_{q}$ for a desired false alarm rate, $\mathcal{A}$. For example, if $r_t$ is Gaussian, $q_t$ would be a chi-squared random variable, and the optimum threshold $\alpha^{\star}_{q}$ corresponding to the desired false alarm rate $\mathcal{A} = \mathcal{A}^{\star}$ is then 
\begin{equation} \label{eqn_chi_squared_threshold}
\alpha^{\star}_{q} = \alpha_{\chi^{2}} := 2 \mathbb{P}^{-1}\left(1-\mathcal{A}^{\star}, \frac{p}{2}\right),
\end{equation}
where $\mathbb{P}^{-1}(\cdot, \cdot)$ denotes the inverse regularized lower incomplete gamma function \cite{murguia2019model}. When the complete distribution is not available, tuning methods using \eqref{eqn_chi_squared_threshold} may design thresholds that generate actual false alarm rates significantly higher than what is desired. With the distributionally robust approach, we aim to achieve a false alarm rate less than $\mathcal{A}$, and the detector threshold $\alpha^{\star}_{r,k}$ is selected such that
\begin{align}  \label{eqn_r_ineq}
\sup_{P_{r} \in \mathcal{P}^{r}_{k}} P_{r}(r^{\top}_{t} \Sigma^{-1}_{r} r_{t} \leq \alpha^{\star}_{r,k}) &= 1 - \mathcal{A}.
\end{align}
When the first two moments $(k = 2)$ of $r_t$ are known, the following proposition using the generalized Chebyshev inequality explained in \cite{boyd_sdp,chen_chebyshev} can be used to obtain the worst case detector threshold $\alpha^{\star}_{r,2}$ satisfying \eqref{eqn_r_ineq}.
\begin{proposition} \label{prop_dr_alfa}
\cite{renganathan2020distributionally} Given a desired false alarm rate $\mathcal{A}$ and $r_t \sim P_{r} \in \mathcal{P}^{r}_{k}$ with $k = 2$, the optimal distributionally robust threshold $\alpha^{\star}_{r,2}$ satisfying \eqref{eqn_r_ineq} is
\begin{align} \label{eqn_DR_alpha}
    \alpha^{\star}_{r,2} = \frac{p}{\mathcal{A}}.
\end{align}
\end{proposition}
\subsection{Improved Detector Threshold With Higher-Order Moments} \label{subsec_k_moments}
It is possible to obtain a sharpened detector threshold than $\alpha^{\star}_{r,2}$, if higher-order moments are taken into account. For example, the skewness and kurtosis parameters convey asymmetry and heaviness of tails of the distribution, respectively. We can leverage such information about the true but unknown distribution revealed by the higher-order moments to tighten the required probability bound and thereby obtain an improved detector threshold. It is possible to use $r_{t}$ with higher order moments (first $k$ moments) to obtain a sharpened detector threshold $\alpha^{\star}_{r,k}$. However, it was shown in \cite{bertsimas_popescu} that for $r_t \in \mathbb{R}^{p}$ with support $\Omega = \mathbb{R}^{p}$ and $k \geq 4$, it is NP-hard to find tight bounds for the corresponding moment bound problem with rational problem data. On the other hand, they provide a semidefinite optimization problem in $k+1$ dimension for the case of a univariate random variable with $k$ moments. Hence, rather looking for higher-order moments of $r_{t}$, instead we look for higher-order moments of scalar random variable $q_{t}$. Subsequently, we use a bisection algorithm to obtain a sharpened detection threshold $\alpha^{\star}_{q,k}$ for a given $\mathcal{A}$.
\subsection{Estimating Probability Using $(1,k,\Omega)-$Moment Bound}
\begin{figure}[ht]
    \centering
    \includegraphics[scale=0.35]{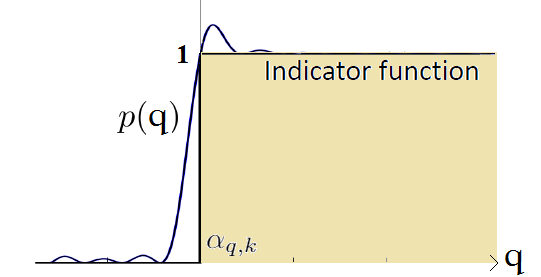}
    \caption{A moment based polynomial bounding the indicator function $\mathbf{1}_{\mathcal{S}_{k}}$ representing the set $\mathcal{S}_{k} = \mathbb{R}_{> \alpha_{q,k}}$ is shown here.}
    \label{fig_poly}
\end{figure}
Given the first $k$ moments of random variable $q_t$ with support $\Omega = \mathbb{R}_{\geq 0}$ and the set $\mathcal{S}_{k} = \mathbb{R}_{> \alpha_{q,k}}$ representing an alarm event as shown in Figure \ref{fig_poly}, the infinite dimensional $(1,k,\Omega)-$moment bound primal problem is given by
\begin{equation} \label{eqn_1komega_primal}
    \begin{aligned}
        &\sup_{P_{q} \in \mathcal{P}^{q}_{k}} & & P_{q}(q_{t} \in \mathcal{S}_{k})\\
        &\text{s.t.} & &\int_{\Omega} q^{r} d\mu = M^{r}_{q}, \quad r = 0,1,\dots,k.
    \end{aligned}
\end{equation}
Since \eqref{eqn_1komega_primal} is infinite dimensional, it is difficult to solve efficiently. However, we can use the linear programming duality theory to associate a dual variable $y_{r}, r = 0,1,\dots,k$ with each equality constraint of \eqref{eqn_1komega_primal} to get the corresponding dual problem
\begin{equation} \label{eqn_1komega_dual}
    \begin{aligned}
        &\text{min} & & \sum^{k}_{r = 0}y_r M^{r}_q\\
        &\text{s.t.} & &p(q) = \sum^{k}_{r = 0}y_r q^r \geq 1, \quad \forall q \in \mathcal{S}_{k},\\ 
        & & &p(q) = \sum^{k}_{r = 0}y_r q^r \geq 0, \quad \forall q \in \Omega.
    \end{aligned}
\end{equation}
The probability obtained as the solution to \eqref{eqn_1komega_dual} is an upper bound to the probability associated with \eqref{eqn_1komega_primal}. In general, if the moment vector $\bar{\sigma}_{k} = (M^{0}_q, M^{1}_q, \dots M^{k}_q)$ is an interior point of the set $\mathcal{M}_{k}$ of all feasible moment vectors, then strong duality exists between \eqref{eqn_1komega_primal} and \eqref{eqn_1komega_dual} enabling us to obtain a \emph{tight} bound on $P_{q}(q_{t} \in \mathcal{S}_{k})$. To achieve a desired false alarm rate $\mathcal{A}$, we can tune the threshold $\alpha_{q,k}$ defining the set $\mathcal{S}_{k}$ such that solution to \eqref{eqn_1komega_dual} is $\mathcal{A}$. The following lemma establishes the general trend observed between the values of the tuned threshold $\alpha_{q,k}$ for increasing values of $k$. 

\begin{lemma} \label{lemma_1}
Let $\mathcal{A}$ be the desired false alarm rate and $q_{t} \sim P_{q} \in \mathcal{P}^{q}_{k}$ as in \eqref{eqn_ambig_set_q_high}. Then, $\forall j,k \in \mathbb{Z}_{\geq 1}$ with $j < k$, we have 
\begin{align}
    \alpha^{\star}_{q,k} \leq \alpha^{\star}_{q,j}
\end{align} 
when the associated moments agree up to order $j$.
\end{lemma}
\begin{proof}
When the associated moments up to order $j<k$ agree, clearly we have $\mathcal{P}^{q}_{k} \subseteq \mathcal{P}^{q}_{j}$, since additional moment constraints restrict the set of distributions.
For a fixed threshold $\alpha^{\star}_{q,j}$ tuned for the false alarm rate $\mathcal{A}$, it follows that
\begin{align}
    \sup_{P_{q} \in \mathcal{P}^{q}_{k}} P_{q}(q_{t} > \alpha^{\star}_{q,j}) \leq \sup_{P_{q} \in \mathcal{P}^{q}_{j}} P_{q}(q_{t} > \alpha^{\star}_{q,j}) = \mathcal{A}.
\end{align}
Thus, to achieve a desired false alarm rate $\mathcal{A}$ under the constraint that $P_{q} \in \mathcal{P}^{q}_{k}$, the threshold $\alpha^{\star}_{q,k}$ must be non-increasing in $k$, satisfying $\alpha^{\star}_{q,k} \leq \alpha^{\star}_{q,j}$ for $j < k$.
\end{proof}

\begin{theorem} \label{thm_alfa_higher_moments}
Let $\epsilon > 0$, $k \in \mathbb{Z}_{\geq 2}$ and consider $q_t$ with its first $k$ moments of distribution $(M^1_q,M^2_q,\dots,M^k_q)$ (we let $M^0_q = 1$) defined on $\mathbb{R}_{+}$ being known (or estimated from the given data) and the associated anomaly detector threshold $\alpha_{q,k}$ introduced in \eqref{eqn_detectorthresh}, which is intended to achieve a desired false alarm rate $\mathcal{A}$. Suppose that $\alpha_{q,k}$ is obtained by solving the following bisection algorithm (knowing that the desired $\alpha^{\star}_{q,k} \in [\alpha_{l}, \alpha_{u}]$ and $\alpha_{u} = \alpha^{\star}_{q,k-1}$ by Lemma \ref{lemma_1})
\begin{algorithm} 
\caption{\texttt{Bisection} Subroutine}
\While{$\alpha_u - \alpha_l > \epsilon$}
{
$\alpha^{\star}_{q,k} \gets (\alpha_u + \alpha_l)/2$ \\
$p_{min} \gets $\text{ optimal value of} \eqref{eqn_sdp_higher_moments} with $\alpha_{q,k} = \alpha^{\star}_{q,k}$ \\
\eIf{$p_{min} > \mathcal{A}$} 
{$\alpha_l \gets \alpha^{\star}_{q,k}$}
{$\alpha_u \gets \alpha^{\star}_{q,k}$}
}
\label{alg_bisect}
\end{algorithm}
where the SDP that gives the tight upper bound on $P_{q}(q_t \geq \alpha_{q,k})$ in the third line is
\begin{equation} \label{eqn_sdp_higher_moments}
    \begin{aligned}
        &\text{min} & & \sum^{k}_{r = 0}y_r M^{r}_{q}\\
        &\text{s.t.} & & \sum_{i,j:i+j=2l-1} x_{ij} = 0, &l = 1,\dots,k, \\ 
        & & & (y_{0}-1) + \sum^{k}_{r = 1}y_r \alpha^{r}_{q,k} = x_{00}, \\
        & & & \sum^{k}_{r = l}y_r \binom{r}{l} \alpha^{r-l}_{q,k} = \sum_{i,j:i+j=2l} x_{ij}, &l = 1,\dots,k, \\
        & & & X \succeq 0, \\
        & & & \sum_{i,j:i+j=2l-1} z_{ij} = 0, &l = 1,\dots,k, \\ 
        & & & \sum^{l}_{r = 0}y_r \binom{k-r}{l-r} \alpha^{r-l}_{q,k} = \sum_{i,j:i+j=2l} z_{ij}, &l = 0,\dots,k, \\
        & & & Z \succeq 0. 
    \end{aligned}
\end{equation}
with variables $X, Z\in \mathbb{R}^{(k+1) \times (k+1)}, y_r, r = 0,1,\dots,k$. Then with this optimal detector threshold $\alpha^{\star}_{q,k}$, the false alarm rate under the worst-case distribution of the $q_t$ specified by the $k$-moments-based ambiguity set \eqref{eqn_ambig_set_q_high} is at most $\mathcal{A}$.
\end{theorem}
\begin{proof}
To get an upper bound on required probability in \eqref{eqn_1komega_primal}, it suffices to check for the polynomials defined in \eqref{eqn_1komega_dual} to be non-negative in their respective sets. That is, the polynomial $p(q)$ satisfies $p(q) - 1 \geq 0, \forall q \in \mathcal{S}_{k}$ if and only if there exists a positive semidefinite matrix $X \in \mathbb{R}^{(k+1) \times (k+1)}$ that satisfies the first four constraints defined in \eqref{eqn_sdp_higher_moments}. Similarly, the same polynomial $p(q)$ satisfies $p(q) \geq 0, \forall q \in [0, \alpha_{q,k}]$ if and only if there exists a positive semidefinite matrix $Z \in \mathbb{R}^{(k+1) \times (k+1)}$ that satisfies the last three constraints defined in \eqref{eqn_sdp_higher_moments}. Then using the results given in Theorem 3.2 of \cite{bertsimas_popescu}, for a fixed threshold $\alpha_{q,k}$, the optimal value of the semidefinite program in \eqref{eqn_sdp_higher_moments} yields the worst-case false alarm rate generated by the worst-case $P_{q} \in \mathcal{P}^{q}_{k}$ in \eqref{eqn_ambig_set_q_high}. To ensure that the optimal value of \eqref{eqn_sdp_higher_moments} is equal to the desired false alarm rate $\mathcal{A}$, it suffices to tune the threshold $\alpha_{q,k}$ using the bisection algorithm \ref{alg_bisect}, to get $\alpha^{\star}_{q,k}$.
\end{proof}

\subsection{Discussion}
For $k = 1,2,3$, the solution to  \eqref{eqn_sdp_higher_moments} can be obtained in closed-form using Theorem 3.3 of \cite{bertsimas_popescu} with the corresponding moments data without explicitly solving the SDP in \eqref{eqn_sdp_higher_moments}. For $k = 1,2$, we recover the Markov bound and a strictly improved Chebyshev bound respectively as the solution of \eqref{eqn_sdp_higher_moments}. Assuming that $\hat{\delta}_{k} > 0, \alpha^{\star}_{q,k} = (1+\hat{\delta}_{k}) M^{1}_{q}$, the resulting closed form solutions for the optimal detector threshold $\alpha^{\star}_{q,k}$ for a desired false alarm rate $\mathcal{A}$ with $k = 1,2$ are summarized in Table \ref{tab_alfas} with $C^{2}_{M} = \frac{M^{2}_{q} - (M^{1}_{q})^{2}}{(M^{1}_{q})^{2}}$ and these values can be used as $\alpha_{u}$ towards computing $\alpha^{\star}_{q,k}$ for any $k >2$. We omit the expression for the threshold $\alpha^{\star}_{q,3}$ for the sake of brevity. 
\begin{table}[ht]
    \centering
    \begin{tabular}{|| c | c ||} 
 \hline
 $k$ & $\alpha^{\star}_{q,k}$ value\\ [0.5ex]
 \hline\hline
 1 & $\frac{M^{1}_{q}}{\mathcal{A}}$ \\ \hline
 2 & $\left(1 + \sqrt{\frac{1 - \mathcal{A}}{\mathcal{A}}} \, C_{M} \right)M^{1}_{q}$ \\ [1ex] 
 \hline
 \end{tabular}
    \caption{Closed-form solutions for $\alpha^{\star}_{q,k}$ with $k = 1,2.$}
    \label{tab_alfas}
\end{table}
Further, the optimal threshold $\alpha^{\star}_{q,k}$ is tight for a given $\mathcal{P}^{q}_{k}$ up to a tolerance $\epsilon > 0$ specified by the bisection in Algorithm \ref{alg_bisect}. 

The exact detector threshold (which we call $\alpha^{\star}_{q}$) corresponding to $P_{q}$ is obtained only asymptotically (or equivalently when true $P_{q}$ is known exactly). That is,
    \begin{align}
        \alpha^{\star}_{q} = \begin{cases} \lim_{k \rightarrow \infty} \alpha^{\star}_{q,k}, &\text{if all moments are known}, \\
        F^{-1}_{q}(\mathcal{A}), &\text{if true $P_{q}, F_{q}$ are known}.
        \end{cases}
    \end{align}
However, while determining how many moments are needed to get a close approximation of the exact threshold is difficult in general, we find that significant improvements can be obtained from a small number of moments. Given $k$ moments of $q_t$, the complexity of this higher-order moment based approach involves solving the SDP given by \eqref{eqn_sdp_higher_moments} with variables $X, Z\in \mathbb{R}^{(k+1) \times (k+1)}, y_r, r = 0,1,\dots,k$. For large $k$, it is advisable to use, e.g., the Legendre polynomial basis instead of the standard polynomial basis for obtaining the moment-based polynomial in \eqref{eqn_sdp_higher_moments} as the former has nice orthogonal properties that improve numerical stability.


\section{Attack-Reachable Set Bounds}\label{sec_reach_sets}
With the attacker assumed to have perfect knowledge of the system dynamics, the Kalman filter, control inputs,  measurements along with read and write access to all the sensors at each time step, we show that the volume of the attack-reachable set shrinks with the tightened threshold. We define a zero-alarm attack, which generates attack sequences so that no alarms are raised during attack. With $\Sigma^{\frac{1}{2}}_{r}$ being the symmetric square root of $\Sigma_{r}$ and an attack input $\bar{\delta}_t$ such that $\bar{\delta}^{\top}_t \bar{\delta}_t \leq \alpha$ (here $\alpha = \alpha^{\star}_{q,k}$), the attack sequence $\delta_t = - Ce_t - v_t + \Sigma^{\frac{1}{2}}_{r} \bar{\delta}_t$ guarantees that no alarm is raised. Then, using a static estimator feedback $u_t = K \hat{x}_t$ with $\bar{\delta}_t$, the evolution of the system dynamics with the joint state $\xi_t=[x_t, e_t]^{\top}$ with input $\zeta_t=[w_t, \bar{\delta}_t]^{\top}$ is studied. We define a reachable set of interest, driven by the ellipsoidally bounded inputs $w_t$ and $\bar{\delta}_t$, as
\begin{equation}
\mathcal{R}_x = \left\{  x_t=[I_n,\, 0_{n\times n}]\xi_t \ \left|\
    \begin{aligned}
        &\xi_{t+1}=\hat{A}\xi_{t}+ \hat{B} \zeta_t, \\
        &\xi_1=\mathbf{0},\ \bar{\delta}^{\top}_t \bar{\delta}_t \leq \alpha,\\
        &w_t^\top\Sigma_w^{-1}w_t \leq \bar{w},\ \forall t\in\mathbb{N}
    \end{aligned}
    \right. 
    \right\},
\end{equation}
where the noise threshold $\bar{w}$ obtained using \eqref{eqn_DR_alpha} satisfies,  
\begin{equation}\label{eqn_DRnoise_threshold}
    \sup_{P_{w} \in \mathcal{P}^{w}_{k}} P_{w}(w^{\top}_{t} \Sigma^{-1}_{w} w_{t} \leq \bar{w}) = 1 - \mathcal{A},
\end{equation}  
and $\hat{A} = \begin{bmatrix}A+BK & -BK \\ 0 & A\end{bmatrix}, \hat{B} = \begin{bmatrix}I & 0\\I & -L\Sigma_r^{1/2}\end{bmatrix}$. 
Using the geometric approach presented in \cite{hashemi2020codesign} with $H_i = A_{cl}^{i} - A^{i}$, the reachable set of states is the Minkowski sum of the following ellipsoidal bound
\begin{align}\label{eqn_reach_minkowski}
\text {\footnotesize 
$\mathcal{R}_{x,t}(\bar{w}, \alpha)  =\bigoplus_{i=0}^{t-2} \mathcal{E}\left(\bar{w} A^i \Sigma_w {A^{i}}^{\top}\right) \oplus \mathcal{E}\left(\alpha H_i L\Sigma_r L^{\top} H_{i}^{\top}\right).$}
\end{align}
Specifically, theorem 1 of \cite{hashemi2020codesign} provides us the exact boundary of Minkowski sum in \eqref{eqn_reach_minkowski} using an analytical formula. The following corollary highlights the effects of the tightened threshold on the size of the reachable set.
\begin{corollary} \label{cor_reach_shrinks}
The volume of the reachable set $\mathcal{R}_{x,t}(\bar{w}, \alpha)$ shrinks with the tightened threshold $\alpha^{\star}_{q,k}$ for all $k \geq 1$. Further, $\forall k_1, k_2 \geq 1$ and $k_1 < k_2$, the following inclusion holds $\mathcal{R}_{x,t}(\bar{w}, \alpha^{\star}_{q,k_{1}}) \supseteq \mathcal{R}_{x,t}(\bar{w}, \alpha^{\star}_{q,k_{2}})$.
\end{corollary}
\begin{proof}
Apply \eqref{eqn_reach_minkowski} with $k_1, k_2 \geq 1$ \& $k_1 < k_2$ along with $\alpha^{\star}_{q,k_{1}} > \alpha^{\star}_{q,k_{2}}$ (by Lemma \ref{lemma_1}), to see that volume$(\mathcal{R}_{x,t}(\bar{w}, \alpha^{\star}_{q,k_{2}})) \leq \text{volume}(\mathcal{R}_{x,t}(\bar{w}, \alpha^{\star}_{q,k_{1}}))$. Inclusion follows from Minkowski sum with $\alpha^{\star}_{q,k_{1}} > \alpha^{\star}_{q,k_{2}}$.
\end{proof}


\section{Numerical Simulation} \label{sec_numerical_simulation}
We consider an empirical system under study with the detector tuned to a false alarm rate $\mathcal{A} = 0.05$ ($5$\%). We demonstrate here simulation results when the uncertainties are zero-mean Gaussian. We compare the size of the reachable set boundary computed using \eqref{eqn_reach_minkowski} for thresholds $\alpha^{\star}_{q,1}, \alpha^{\star}_{q,2}, \alpha^{\star}_{q,4}$. We assume that the modeler is unaware of the functional form of the uncertainties and has to arrive at a detector threshold satisfying the desired false alarm rate. The SDP in \eqref{eqn_sdp_higher_moments} was solved with $\epsilon = 10^{-4}$ using SOSToolbox in Matlab with SeDuMi solver.  
\begin{equation*}
    \begin{aligned}
        A &= \begin{bmatrix} 0.84 & 0.23 \\ -0.47 & 0.12 \end{bmatrix}, B = \begin{bmatrix} 0.07 & -0.32 \\ 0.23 & 0.58\end{bmatrix}, C = \begin{bmatrix} 1 & 0 \\ 2 & 1 \end{bmatrix} \\
        K &= \begin{bmatrix} 1.404 & -1.402 \\ 1.842 & 1.008 \end{bmatrix}, L = \begin{bmatrix} 0.0276 & 0.0448 \\ -0.01998 & -0.0290 \end{bmatrix}, \\
        \Sigma_v &= 6 I_p, \Sigma_w = \begin{bmatrix} 0.0225 & -0.0055 \\ -0.0055 & 0.0100 \end{bmatrix}.
    \end{aligned}
\end{equation*}  
\begin{figure}
    \centering
    \includegraphics[scale=0.18]{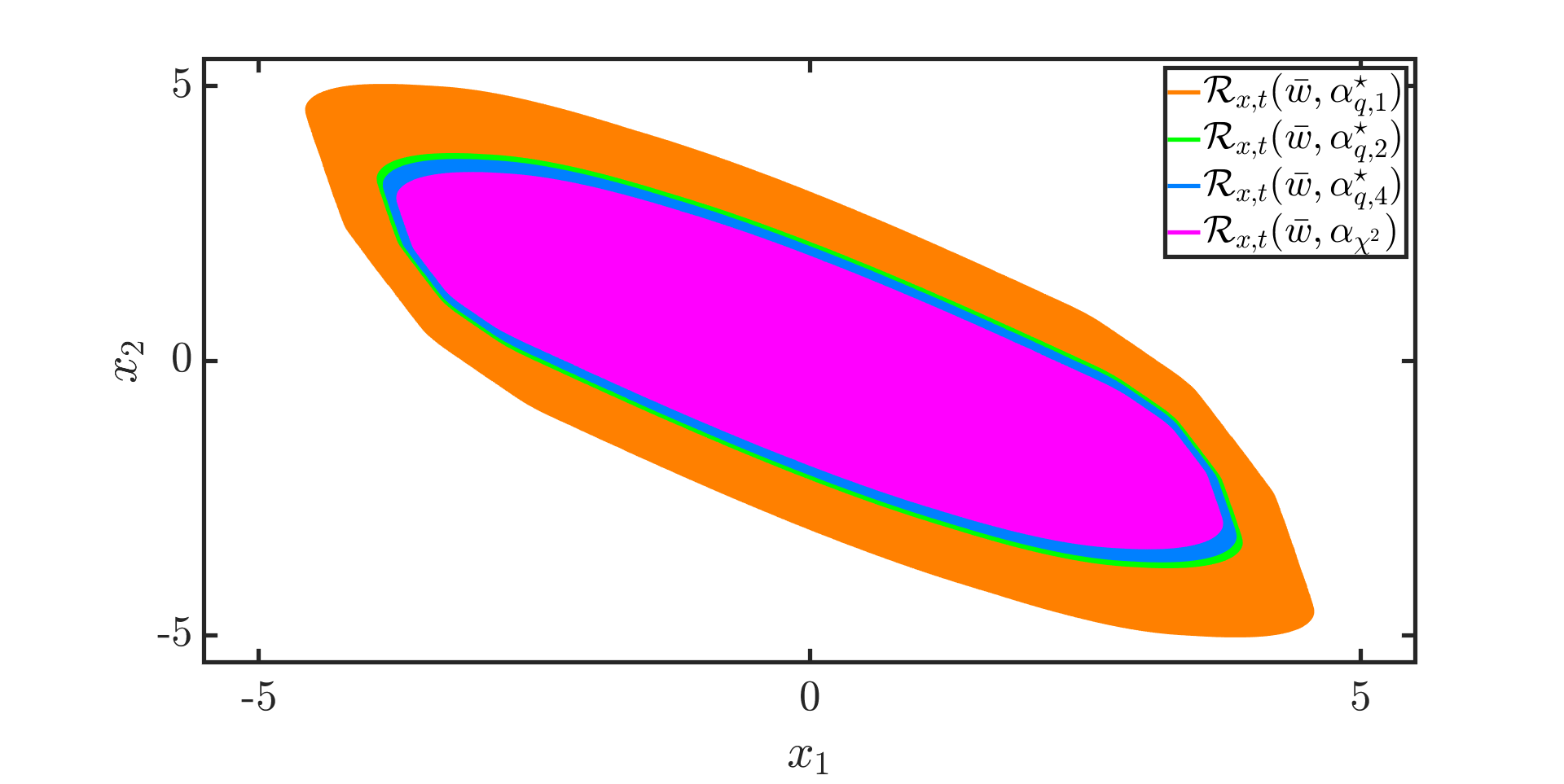}
    \caption{The reachable sets associated with the detector thresholds $\alpha^{\star}_{q,1}, \alpha^{\star}_{q,2}, \alpha^{\star}_{q,4}, \alpha_{\chi^{2}}$ are shown in orange, green, blue and magenta colors respectively. It is evident from the size of the reachable set that given a desired false alarm rate $\mathcal{A}$, the knowledge of higher-order moments results in a tightened detector threshold and thereby restricting the attacker's ability to launch a larger attack.}
    \label{fig_sim_normal_reach}
\end{figure}
\begin{figure}
    \centering
    \includegraphics[scale=0.18]{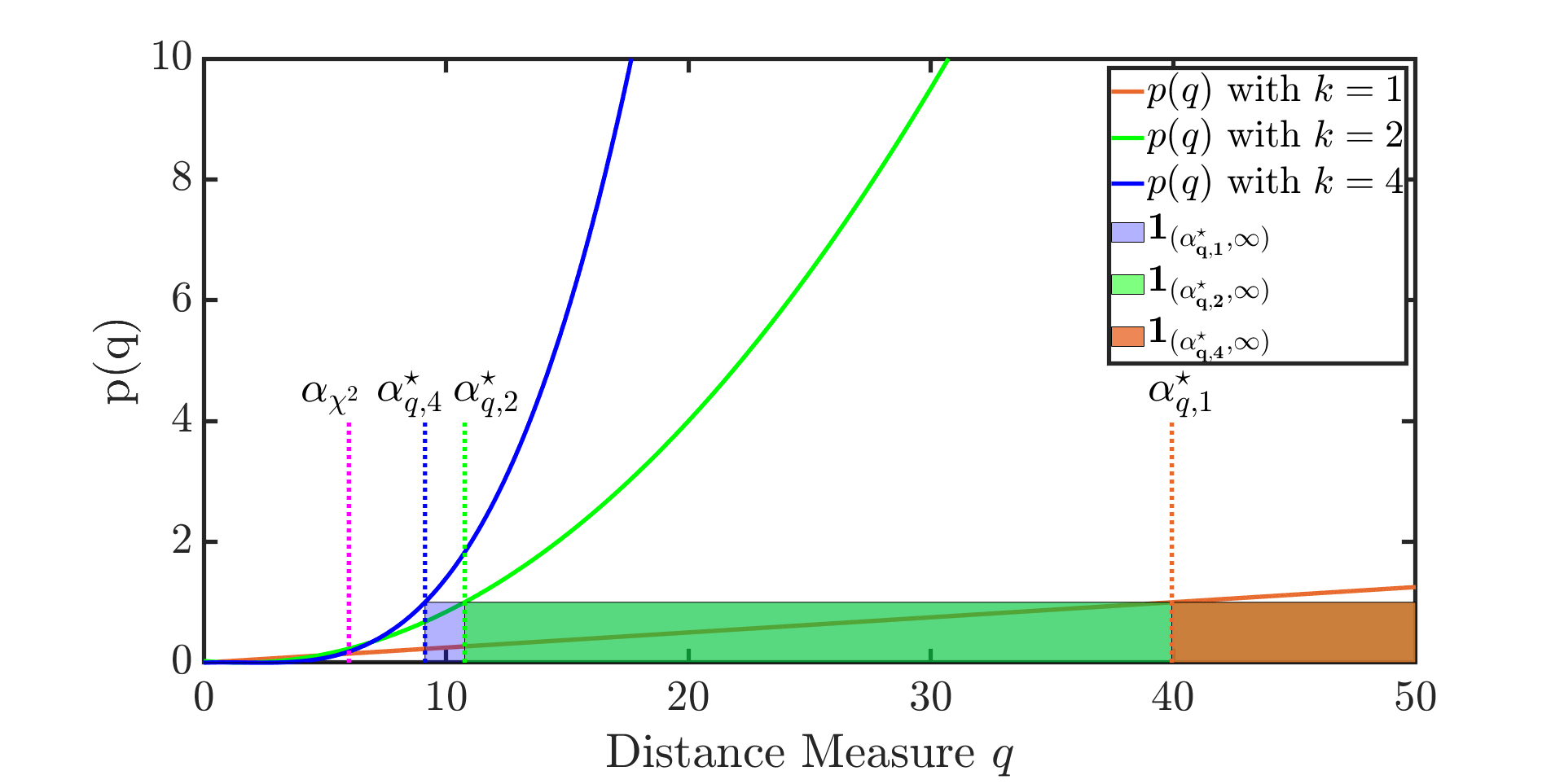}
    \caption{The moment based polynomials in orange, green and blue bounding their respective indicator functions  $\mathbf{1}_{(\alpha^{\star}_{q,1},\infty)}, \mathbf{1}_{(\alpha^{\star}_{q,2},\infty)}, \mathbf{1}_{(\alpha^{\star}_{q,4},\infty)}$ in shaded orange, green and blue colors are shown here. Clearly, $\alpha^{\star}_{q,1}$ is very conservative and with $k = 4$, the threshold $\alpha^{\star}_{q,4}$ starts getting closer to the true threshold $\alpha_{\chi^{2}}$ given by \eqref{eqn_chi_squared_threshold}.}
    \label{fig_sim_normal_indicate}
\end{figure}

When the noises $w_t$ and $v_t$ are truly Gaussian, it is evident from Fig. \ref{fig_sim_normal_reach} that the reachable set corresponding to the thresholds $\alpha^{\star}_{q,1} = \alpha^{\star}_{r,2} = 40$ is conservative and ensures that the false alarm does not exceed 5\% but this also provides the attacker with the ability to launch a larger attack. However, with the knowledge of additional moments, the detector thresholds $\alpha^{\star}_{q,2} = 10.7684$, $\alpha^{\star}_{q,4} = 9.1315$ get tightened with $\alpha^{\star}_{q,4} \leq \alpha^{\star}_{q,2} \leq \alpha^{\star}_{q,1}$ as shown in Fig. \ref{fig_sim_normal_indicate}. Further, this threshold tightening limits the attacker's ability to launch a larger attack which is depicted through the reachable sets corresponding to the thresholds $\alpha^{\star}_{q,1}, \alpha^{\star}_{q,2}, \alpha^{\star}_{q,4}$ as shown in Fig. \ref{fig_sim_normal_reach}. Subsequently, the false alarm rate corresponding to the threshold $\alpha_{\chi^{2}}$ was 5\% as expected and with the thresholds $\alpha^{\star}_{q,4}, \alpha^{\star}_{q,2}, \alpha^{\star}_{q,1}$, it dropped to $1\%, 0.45\%, 0\%$ respectively. When noises are truly multi-variate Laplacian (which has heavier tails than normal distribution with same mean and covariance), it resulted in thresholds $\alpha^{\star}_{q,1} = 39.83$, $\alpha^{\star}_{q,2} = 17.23, \alpha^{\star}_{q,4} = 16.54$ with false alarm rates $0\%, 0.9\%, 1\%$ respectively. Thus, no matter what distributions satisfying \eqref{eqn_ambig_set_w_high}, \eqref{eqn_ambig_set_v_high} govern the noises $w_t$, $v_t$ respectively, the inclusion of higher-order moments restricts the attacker's potential impact through a tightened detector threshold. 

\section{Conclusion \& Future Outlook}\label{sec_conclusion} 
We have proposed a distributionally robust approach to form the $k$ moments based ambiguity set for the detection measure data and used it to tune the anomaly detectors for a desired false alarm rate. We found a detector threshold which guaranteed that the false alarm rate did not exceed a desired value using a semidefinite program. We have demonstrated the effectiveness of our proposed approach with a numerical example. Further, our approach using higher-order moments restricted the attacker's potential impact. Future works include addressing the problems associated with the data-driven formulation with moment estimation uncertainty and securing nonlinear cyberphysical systems with distributionally robust unscented Kalman filter based state estimation.   


\bibliographystyle{IEEEtran}
\bibliography{security}

\end{document}